\documentclass[11pt]{article}
\usepackage{graphicx}
\usepackage{longtable,lscape}
\usepackage{amsthm}
\usepackage{amsfonts}
\usepackage{amsmath}
\usepackage{bbm}
\usepackage{float}
\usepackage{url}
\newtheorem{definition}{Definition}
\newtheorem{theorem}{Theorem}

\newcommand{\captionfonts}{\footnotesize}
\makeatletter  % Allow the use of @ in command names
\long\def\@makecaption#1#2{%
  \vskip\abovecaptionskip
  \sbox\@tempboxa{{\captionfonts #1: #2}}%
  \ifdim \wd\@tempboxa >\hsize
    {\captionfonts #1: #2\par}
  \else
    \hbox to\hsize{\hfil\box\@tempboxa\hfil}%
  \fi
  \vskip\belowcaptionskip}
\makeatother 
\begin{document}
\title{Quantum Model Theory (QMod): \\ Modeling Contextual Emergent Entangled Interfering Entities}
\author{Diederik Aerts and Sandro Sozzo \vspace{0.5 cm} \\ 
        \normalsize\itshape
        Center Leo Apostel for Interdisciplinary Studies \\
        \normalsize\itshape
        Brussels Free University \\ 
        \normalsize\itshape
         Krijgskundestraat 33, 1160 Brussels, Belgium \\
        \normalsize
        E-Mails: \url{diraerts@vub.ac.be,ssozzo@vub.ac.be} \\
               }
\date{}
\maketitle
\begin{abstract}
\noindent
In this paper we present \emph{Quantum Model Theory} (\emph{QMod}), a theory we developed to model entities that entail the typical quantum effects of {\it contextuality}, {\it superposition}, {\it interference}, {\it entanglement} and {\it emergence}. The aim of \emph{QMod} is to put forward a theoretical framework that 
is more general than standard quantum mechanics, in the sense that, for its complex version it only uses this quantum calculus locally, i.e. for each context corresponding to a measurement, and for its real version it does not need the property of `linearity of the set of states' to model the quantum effect. In this sense, \emph{QMod} is a generalization of quantum mechanics, similar to how the general relativity manifold mathematical formalism is a generalization of special relativity. 
We prove by means of a representation theorem that \emph{QMod} can be used for any entity entailing the typical quantum effects mentioned above. Some examples of application of QMod in concept theory and macroscopic physics are also considered.
\end{abstract}

\medskip
{\bf Keywords}: Quantum modeling, contextuality, interference, QMod

\section{Introduction\label{intro}}
Over the years it has become clear that quantum structures do not only appear within situations in the micro world, but 
also arise in 
the macro world \cite{aerts1982,aerts1986,aertsdurtgribvanbogaertzapatrin1993,aerts2009,aerts2009b,aerts2010,aerts2010b}. In this respect, more recently \cite{aerts2009,aerts2009b,aerts2010,aerts2010b}, four major effects have been put forward  
which also appear in macroscopic situations, and give rise to the presence of quantum structures. These effects are `interference', `contextuality', `emergence' and `entanglement'. Sometimes it has been possible to use the full quantum apparatus of linear operators in complex Hilbert space to model these effects as they appear in macroscopic situations. But, in quite some occasions a formalism more general than standard quantum theory in complex Hilbert space is needed.

When investigating the structure of concepts, and how they combine to form sentences and texts, we already proposed a generalization of the standard quantum formalism, which we called a \emph{State Context Property}, or \emph{SCoP}, formalism, specifically designed to model concepts and their combinations \cite{aertsgabora2005a}. This generalization was inspired by work on quantum axiomatics \cite{aerts2002}, and later also used to analyze aspects of concepts, or inspire contextual approaches \cite{aertsgabora2002,nelson2007,gaboraroschaerts2008,flenderkittobruza2009,gaboraaerts2009,dhooghe2010,aertsczachorsozzo2010,velozgaboraeyjolfsonaerts2011,aertssozzo2011}. However, the SCoP formalism is %a 
very general, which led us to reflect about a formalism closer to the complex Hilbert space of standard quantum theory, more specific and hence mathematically more
efficient than SCoP, but at the same time general enough to cope with the modeling of the main quantum effects identified in the domains different from the micro-world.

In this article we aim to propose a general modeling theory capable of modeling situations in which the effects of contextuality, emergence, entanglement and interference appear. We will call this modeling theory \emph{Quantum Model Theory}, or {\it QMod}. It is not just a broad modeling scheme, because a specific powerful mathematical representation theorem is the heart of it. It is due to this representation theorem that the mathematical structure of {\it QMod} contains the potential to describe entanglement and emergence. We will see that the standard quantum mechanical formalism is a special case of {\it QMod}, were emergence and entanglement are consequences of respectively the linear structure of the Hilbert space and the tensorproduct procedure for compound quantum systems. In {\it QMod} no linearity is needed in principle, although it can be introduced if useful. This is the fundamental reason why {\it QMod} constitutes a powerful and helpful generalization of the quantum formalism. We will also see that {\it QMod} is a concretization of SCoP in a specific way when it is used to model concepts and their combinations. However, {\it QMod} is not only meant to model concepts and their combinations, like it was the case for SCoP. It aims at modeling all situations of entities where the effects of interference, contextuality, emergence and entanglement play a role, and in this sense it is more general than SCoP in its applications. 

Proceeding by analogy, we could say that {\it QMod} is a generalization of classical and quantum theory in a very similar way to how the general relativistic manifold formalism is a generalization of special relativity. 
Indeed, general relativity theory assumes that, for each point of space-time, hence locally, space-time can be considered operationally in an Euclidean way. Similarly, {\it QMod} assumes that, whenever a given measurement is considered, hence again locally, the probabilities are defined operationally and locally for this one measurement, and for an arbitrary set of states of the considered entity. For one fixed state and when an event space is defined on the set of outcomes, one can, for example, assume Kolmogorov's axioms to be valid. Remark that since an arbitrary set of states is considered locally, the probability structure will not be a single Kolmogorovian probability space, but a set of Kolmogorovian models, one for each state. Of course, we expect the overall general probability model to be non-Kolmogorovian in {\it QMod} if different measurements and different states are considered, as we expect that a non-Euclidean model arises if different points of space-time are considered in general relativity theory.

For the sake of completeness, let us briefly resume the content of this paper. We introduce the essentials of {\it QMod} in Sec. \ref{sectionrepresentation}, where we prove a representation theorem providing the steps needed to construct a quantum modeling. More specifically, the representation theorem states that any entity that can be described by means of its states, its contexts and specifically corresponding operationally defined probabilities admits a mathematical modeling in $\mathbb{R}^{n}$ and $\mathbb{C}^{m}$, with $n$ and $m$ suitably chosen. Successively, we supply in Sec. \ref{appl} applications that illustrate how {\it QMod} concretely works. Indeed, we firstly consider a conceptual entity (Sec. \ref{concepts}), then a macroscopic entity (Sec. \ref{vessel}). These two examples admit real and complex representations in ${\mathbb R}^{2}$ and ${\mathbb C}^{2}$, respectively. Finally, an abstract example in three dimensions (Sec. \ref{3example}) is provided which admits a real representation in ${\mathbb R}^{3}$ and a complex representation in ${\mathbb C}^{3}$. The treatment of a specific quantum effect, namely interference, concludes the paper (Sec. \ref{interference}). The latter allows one to understand how a mixed state differs from a superposition state within {\it QMod} \cite{young1802,debroglie1923,schrodinger1926,debroglie1928,jonsson1961,feynman1965,jonsson1974,ArndtNairzVos-AndreaeKellervanderZouwZeilinger1999}.

\section{A representation theorem\label{sectionrepresentation}}
Let us introduce the fundamental notions we need for our purposes. We will model `entities' where the notion of entity is to be understood in its most general sense. An entity is a collection of aspects of reality that hang together in such a way that different states exist without loosing the possibility of identification of the same entity in each of these states. Sometimes only one state exists, this is then the limiting case, and the entity is then just a situation. Let us right away give two examples that we will be using in the course of this article to illustrate different aspects of the quantum modeling scheme we introduce.

A first example of an entity that we consider is the concept {\it Animal}. The concept {\it Animal} can indeed be in different states, for example {\it Ferocious Animal} and {\it Sweet Animal} are two such possible states, and many more exists. Each time an adjective is put in front of the concept {\it Animal} another state of {\it Animal} is realized. But also each sentence, or paragraph, or piece of text, surrounding the concept {\it Animal}, places it in a different state. Also exemplars of the concept {\it Animal}, such as {\it Horse} and {\it Bear} can be considered to be states of the concept {\it Animal}. So clearly a large set of different states exist for the concept {\it Animal}. A second example of an entity that we consider is a {\it Vessel of Water}. Different volumes of water that the vessel can contain are different states of the {\it Vessel of Water}.

Along with the notions of entity and state we introduce the notions of measurement and outcome. A measurement consists of a specific context that is realized for the entity being able to be in different states, which is the reason that in some cases, e.g. in SCoP, we use the term `context' to indicate the measurement or measurement context. This context affects generally the state of the entity in different ways and as a consequence different outcomes to the measurement defined by this context can occur for a specific state of the entity. Usually the state of the entity is changed by the measurement, and the resulting state after the measurement on the entity can be identified with each of the outcomes. If this is the case, such an outcome can also be represented by this state. This is in fact why we did not introduce the notion of outcome in SCoP, because contexts, i.e. measurements in SCoP, always change a state of a concept into different possible new states, and hence the outcomes of this measurement context are identified with these states in the case of SCoP. Now, we explicitly want to introduce the notion of outcome in {\it QMod}, because we also want to be able to model situations where the state after a measurement is not identified. We also introduce the notion of probability of occurrence of an outcome, a measurement being performed, with the entity being in a state, as the limit of the relative frequency of this outcome by repetition. It is for this limit of relative frequency for a fixed state, after defining events related to the outcome set, for which Kolmogorov's axioms of probability can be supposed to be valid. In the case of SCoP, this probability is a transition probability from the state before the measurement context to the state after the measurement context.

\begin{definition}[Entity, State, Measurement, Outcome, Probability]
We consider the situation of an entity $S$ that can be in different states, and denote states by $p, q, \ldots$, and the set of states by $\Sigma$. Different measurements can be performed on the entity $S$ being in one of its states, and we denote measurements by $e, f, \ldots$, and the set of measurements by ${\cal M}$. With a measurement $e \in {\cal M}$ and the entity in state $p$, corresponds a set of possible outcomes $\{x_1, x_2, \ldots, x_j, \ldots, x_n\}$, and a set of probabilities $\{\mu(x_j,e,p)\}$, where $\mu(x_j,e,p)$ is the limit of the relative frequency of the outcome $x_j$, the situation being repeated where measurement $e$ is performed and the entity $S$ is in state $p$. We denote the final state corresponding to the outcome $x_j$ by means of $p_j$.
\end{definition}
The following theorem proves that it is always possible to realize the above introduced situation by means of a specific mathematical structure making use of a space of real numbers where the probabilities are derived as Lebesgue measures of subsets of real numbers. We also prove that on top of this real number realization a complex number realization exists, where the probabilities are calculated by making use of a scalar product similar to the one used in the quantum formalism.

\begin{theorem}[Representation theorem] \label{theoremrepresentation}
Let us consider a measurement $e \in {\cal M}$ and a state $p \in \Sigma$, and introduce the set $\{\mu(x_j, e, p) \ | \ j=1,\ldots, n \}$ of probabilities, where $\{x_1, \ldots, x_j, \ldots, x_n\}$ is the set of possible outcomes given $e$ and $p$. Then, it is possible to work out a representation of this situation in $\mathbb{R}^n$ where the probabilities are given by Lebesgue measures of appropriately defined subsets of $\mathbb{R}^n$, and a representation in $\mathbb{C}^m$ where the measurement is modeled in an analogous way as this is the case in the mathematical formalism of standard quantum theory defined on $\mathbb{C}^m$ as a complex Hilbert space.
\end{theorem}
\begin{proof}
We introduce the space $\mathbb{R}^n$, and its canonical basis $h_1=(1, \ldots, 0, \ldots, 0)$, $h_2=(0, 1, 0, \ldots, 0)$, \ldots, $h_j=(0, \ldots, 1, \ldots)$, \ldots, $h_n=(0, \ldots, 1)$. The situation of the measurement $e$ and state $p$ can be represented by the vector
\begin{equation}
v(e,p)=\sum_{j=1}^n\mu(x_j,e,p)h_j
\end{equation}
which is a point of the simplex $S_n(e)$, the convex closure of the canonical basis $\{h_1, \ldots, h_j, \ldots, h_n\}$ in $\mathbb{R}^n$. We call $A_j(e,p)$ the convex closure of the vectors $\{h_1, h_2, \ldots$ $, h_{j-1}, v(e,p), h_{j+1}, \ldots, h_n\}$. We use this configuration to construct a micro-dynamical model for the measurement dynamics of $e$ for the entity in state $p$. This micro-dynamics is defined as follows, a vector $\lambda$ contained in the simplex $S_n(e)$, hence we have
\begin{equation}
\lambda=\sum_{j=1}^n\lambda_jh_j \quad 0 \le \lambda_j \le 1 \quad \sum_{j=1}^n\lambda_j=1
\end{equation}
determines the dynamics of the measurement $e$ on the state $p$ in the following way. If $\lambda \in A_j(e,p)$, and is not one of the boundary points (hence $\lambda$ is contained in the interior of $A_j(e,p)$), then the measurement $e$ gives with certainty, hence deterministically, rise to the outcome $x_j$, with the entity being in state $p$. If $\lambda$ is a point of the boundary of $A_j(e,p)$, then the outcome of the experiment $e$, the entity being in state $p$, is not determined. Let us prove that from the above construction we can derive the probabilities $\mu(x_j, e, p)$ from just Lebesgue measuring the sets of relevant real numbers as subsets of $S_n(e)$. Of course, we make the hypothesis that the micro-dynamical modeling of the measuring process is such that the vector $\lambda$ is chosen at random in the simplex $S_n(e)$ with a randomness modeled by the Lebesque measure on this simplex. Then, following the formulation of the micro-dynamics of the measurement process $e$ for $S$ being in state $p$, we have that the $\mu(x_j, e, p)$, being the probability to obtain outcome $x_j$, is given by the Lebesgue measure of the set of vectors $\lambda$ that are such that this outcome is obtained deterministically, hence this are the $\lambda$ contained in $A_j(e,p)$, divided by the Lebesgue measure of the total set of vectors $\lambda$, which are the $\lambda$ contained in $S_n(e)$. This means that
\begin{equation} \label{probabilitylebesgue}
\mu(x_j, e, p)={m(A_j(e,p)) \over m(S_n(e))} .
\end{equation}
To calculate the Lebesgue measures, let us introduce the following notations. If $h_1,h_2,\ldots,h_n$ are vectors in $\mathbb{R}^n$, we denote by $M(h_1,h_2,\ldots,h_n)$ the $n\times n$ matrix, where $M_{jk} = (h_j)_k$. We denote by $\det(h_1,h_2,\ldots,h_n)$ the determinant of this matrix $M(h_1,h_2,\ldots,h_n)$, and by ${\it Par}(h_1,h_2,\ldots,h_n)$ the parallelepiped spanned by the $n$ vectors. If we consider the two parallelepipeds ${\it Par}(h_1,h_2,\ldots,h_n)$ and ${\it Par}(h_1,h_2,\ldots, h_{j-1},v(e,p),h_{j+1},\ldots,h_n)$, then they are constructions with the same heights over bases which are the simplexes $S_n(e)$ and $A_j(e,p)$. This means that the volumes of these parallelepipeds, as $n$ dimensional subsets of $\mathbb{R}^n$, are equal to the volumes of the simplexes $S_n(e)$ and $A_j(e,p)$, multiplied by the same constant number $c(n)$, which is a number depending on the global dimension $n$. Now, the volumes of the two parallelepipeds, let us denote them $m({\it Par}(h_1,h_2,\ldots,h_n))$ and $m({\it Par}(h_1,h_2,\ldots, h_{j-1},v(e,p),h_{j+1},\ldots,h_n))$ can be calculated by means of the determinants of their matrices, and hence we can also calculate the volumes of the simplexes by these determinants. More specifically we have
\begin{eqnarray} \label{parallelepiped01}
&&m(S_n(e))=c(n)m({\it Par}(h_1,h_2,\ldots,h_n)) \\ \label{determinant01}
&&m({\it Par}(h_1,h_2,\ldots,h_n))=\det\left( \begin{array}{ccccc}
h_1 & \ldots & h_j & \ldots & h_n
\end{array} \right)  \\ \label{parallelepiped02}
&&m(A_j(e,p))=c(n)m({\it Par}(h_1,h_2,\ldots,v(e,p),\ldots,h_n)) \\ \label{determinant02}
&&m({\it Par}(h_1,h_2,\ldots,v(e,p),\ldots,h_n))=\det\left( \begin{array}{ccccc}
h_1 & \ldots & v(e,p) & \ldots & h_n
\end{array} \right)
\end{eqnarray}
and calculating the determinants of the matrices we get
\begin{eqnarray} \label{lebesgue01}
&&\det\left( \begin{array}{ccccc}
h_1 & \ldots & h_j & \ldots & h_n
\end{array} \right)=1 \\ \label{lebesgue02}
&&\det\left( \begin{array}{ccccc}
h_1 & \ldots & v(e,p) & \ldots & h_n
\end{array} \right)=\mu(x_j, e, p) .
\end{eqnarray} 
From (\ref{parallelepiped01}) and (\ref{parallelepiped02}) follows that
\begin{equation}
{m(A_j(e,p)) \over m(S_n(e))}={m({\it Par}(h_1,h_2,\ldots,v(e,p),\ldots,h_n)) \over m({\it Par}(h_1,h_2,\ldots,h_n))} .
\end{equation}
And from (\ref{determinant01}), (\ref{determinant02}), (\ref{lebesgue01}) and (\ref{lebesgue02}) follows (\ref{probabilitylebesgue}).

For the quantum representation we introduce a set of orthogonal projection operators $\{M_k\ \vert k= 1,\ldots,n\}$ on a complex Hilbert $\mathbb{C}^m$ space, with $n \le m \le n^2$, that form a spectral family. This means that $M_k \perp M_l$ for $k\not=l$ and $\sum_{k=1}^nM_k=\mathbbmss{1}$, and we take the $M_k$ such that they are diagonal matrices in $\mathbb{C}^m$. More concretely, each $M_k$ is a matrix with 1's at some of the diagonal places, and zero's everywhere else. The number of 1's is between 1 and $n$, for each $M_k$, and the collections of 1's hang together, their mutual intersections being empty, and the union of all of them being equal to the collection of 1's of the unit matrix $\mathbbmss{1}$. The state is represented by a vector $w(e,p)$ of $\mathbb{C}^m$, such that
\begin{equation} \label{quantumsolution}
\mu(x_k, e, p)=\langle w(e,p)\ \vert M_k\ \vert w(e,p)\rangle =\parallel M_{k} |w(e,p)\rangle \parallel^{2} .
\end{equation}   
A possible solution is 
\begin{equation}
w(e,p)=\sum_{j=1}^m a_je^{i\alpha(e,p)_j}h_j \quad{\rm with} \quad a_j={1 \over b}\sqrt{\mu(x_j, e, p)}
\end{equation}
where $h_j$ is the canonical basis of $\mathbb{C}^m$, and $b$ is the dimension of the projector $M_k$ if $h_j$ is such that $M_kh_j=h_j$. But this is not the only solution, and it might also not be the appropriate solution for the situation we want to model. It shows however that a solution exists, which proves that it is always possible to built this local quantum model.
\end{proof}

The above theorem is an application of the `hidden measurement' approach that we elaborated in our Brussels research group during the eighties and nineties of the foregoing century, with the aim of formulating a contextual hidden variable model for quantum theory \cite{aerts1986,aerts1993,aerts1994,aertsaerts1997,aertsaertscoeckedhooghedurtvalckenborgh1997,aertssven2002,aertssven2005}.

With the above theorem we have constructed a representation of the collection of states and experiments that lead to the same set of outcomes. In this sense, the $\mathbb{R}^n$ model and the $\mathbb{C}^m$ that we have constructed is a model for the interaction between state and experiment. The set of outcomes constitutes a context in which this interaction takes place. 

Concluding this section, it is important to observe that the representation theorem proved above allows one to identify some quantum--like aspects without the necessity of assuming an underlying linear structure. This aspect will manifestly emerge from the treatment of entanglement in a forthcoming paper \cite{aertssozzo2012b}, where we prove that the tensorproduct structure appears already on the level of the real space description, and that it is possible to identify entangled states of an entity without the need of linearity. Although we mentioned also the effect of `emergence' as one of the characteristic quantum effect, we do not consider `emergence' here or in \cite{aertssozzo2012b}, but give it explicitly attention in \cite{aertsgaborasozzo2012}. One could say that {\it QMod} is a generalization of standard quantum mechanics in the sense that, when the real space representation is used, no linearity at all is at play, and when the complex space representation is used, linearity is present only locally. We do not insist on this point, for the sake of brevity, and refer to \cite{aertssozzo2012b} for a detailed analysis of the linearity issue.  

\section{Applications of {\it QMod}\label{appl}}
{\it QMod} can be applied to the modeling of any type of entity that can be described by a set of states, a set of contexts and probabilities defined for outcomes. In the following, we consider some relevant examples that show how our construction works. As we have anticipated at the end of the previous section, these examples will be employed in the description of entanglement in {\it QMod} in a forthcoming paper \cite{aertssozzo2012b}.

\subsection{Concepts\label{concepts}}
Let us consider the example of the entity which is the concept {\it Animal}. We consider a measurement $e$, where a person is asked to choose between the animal being a {\it Horse} or a {\it Bear}, hence there are two outcomes $\{H,B\}$. We consider only one state for {\it Animal}, namely the ground state which is the state where animal is just animal, i.e. the bare concept, and let us denote it $p$. Let us denote by $\mu(H,e,p)$ the probability that {\it Horse} is chosen when $e$ is performed, and by $\mu(B,e,p)$ the probability that {\it Bear} is chosen.

Let us now work out a mathematical construction put forward in the representation theorem proven in Sec. \ref{sectionrepresentation}. For the measurement $e$ we consider the vector space $\mathbb{R}^2$ and its canonical basis $\{(1, 0), (0, 1)\}$. The state $p$ is contextually represented with respects to the measurement $e$ by the vector $v(e,p)=(\mu(H,e,p),\mu(B,e,p))$ in $\mathbb{R}^2$. We introduce the vector $\lambda=(r, 1-r)$, with $0 \le r \le 1$, such that for $(r, 1-r)$ contained in the convex closure of $(1, 0)$ and $(\mu(H,e,p),\mu(B,e,p))$, we get outcome {\it Bear}, while for $(r, 1-r)$ contained in the convex closure of $(\mu(H,e,p),\mu(B,e,p))$ and $(0, 1)$ we get {\it Horse}. Let us calculate the respective lengths and see that we re-obtain the correct probabilities. Denoting the length of the piece of line from $(1, 0)$ to $(\mu(H,e,p),\mu(B,e,p))$ by $d$, we have ${d \over \sqrt{2}}=\mu(B,e,p)$, and ${\sqrt{2}-d \over \sqrt{2}}=\mu(H,e,p)$. 

We can also construct a quantum mathematics model in $\mathbb{C}^2$. Therefore we consider the vector $w(e,p)=(\sqrt{\mu(H,e,p)}e^{i\alpha(e,p)_H},\sqrt{\mu(B,e,p)}e^{i\alpha(e,p)_B})$ in $\mathbb{C}^2$. We have $\mu(H,e,p)=|\langle (1,0)| w(e,p)\rangle|^2$ and $\mu(B,e,p)=|\langle (0,1)| w(e,p)\rangle|^2$, which shows that also the $\mathbb{C}^2$ construction gives rise to the correct probabilities. 

\subsection{Vessel of water\label{vessel}} 
As a second example, we consider an entity $S$ that is a vessel of water containing a volume of water between 0 and 20 liters. Suppose that we are in a situation where we lack knowledge about the exact volume contained in the vessel, and call $p$ the state describing this situation. We consider a measurement $e$ for the vessel that consists in pouring out the water by means of a siphon, collecting it in a reference vessel, where we can read of the volume of collected water. We attribute outcome $M$ if the volume is more than 10 liters and the outcome $L$ if it is less than 10 liters, hence the set of outcomes for $e$ is $\{ M, L \}$. We introduce the probabilities $\mu(M,e,p)$ and $\mu(L,e,p)$ for the outcomes $M$ and $L$, respectively. As in the case of concepts, we construct a mathematical representation in $\mathbb{R}^2$ and its canonical basis $\{(1, 0), (0, 1)\}$. The state $p$ is contextually represented with respects to the measurement $e$ by the vector $v(e,p)=(\mu(M,e,p),\mu(L,e,p))$ in $\mathbb{R}^2$. The simplex $A_M(e,p)$ is the line connecting the points $(\mu(M,e,p),\mu(L,e,p))$ and $(0,1)$, while the simplex $A_L(e,p)$ is the line connecting the points $(1,0)$ and $(\mu(M,e,p),\mu(L,e,p))$. We introduce the vector $\lambda=(r, 1-r)$, with $0 \le r \le 1$, such that for $(r, 1-r)$ contained in the convex closure of $(1, 0)$ and $(\mu(M,e,p),\mu(L,e,p))$, we get outcome $L$, while for $(r, 1-r)$ contained in the convex closure of $(\mu(M,e,p),\mu(L,e,p))$ and $(0, 1)$ we get $M$. Let us calculate the respective lengths and see that we find back the correct probabilities. Denoting the length of the piece of line from $(1, 0)$ to $(\mu(M,e,p),\mu(L,e,p))=(1/2,1/2)$ by $d$, we have ${d \over \sqrt{2}}=\mu(L,e,p)$, an ${\sqrt{2}-d \over \sqrt{2}}=\mu(M,e,p)$. Thus, $d=\frac{\sqrt{2}}{2}$ allows one to recover the right probabilities. 

The quantum mathematics model in $\mathbb{C}^2$ can be constructed as follows. We consider the orthogonal projection operators 
$M_{M}=\left (\begin{array}{cc} 1 & 0 \\ 0 & 0
 \end{array}\right )$ and  $M_{L}=\left (\begin{array}{cc} 0 & 0 \\ 0 & 1
 \end{array}\right )$, and the vector $w(e,p)=(\sqrt{\mu(M,e,p)}e^{i\alpha(e,p)_M},$ $\sqrt{\mu(L,e,p)}e^{i\alpha(e,p)_L})$ in $\mathbb{C}^2$. We have $\mu(M,e,p)=\langle w(e,p) | M_{M} | w(e,p)\rangle$ and $\mu(L,e,p)=\langle w(e,p) | M_{L} | w(e,p)\rangle$, which also gives rise to the correct probabilities.

\subsection{Illustration in three dimensions\label{3example}}
Let $S$ be an entity and let us consider the situation where the measurement $e$ on $S$ has three possible outcomes $\{x_1, x_2, x_3\}$. We denote by $\mu(x_1, e, p)$, $\mu(x_2, e, p)$ and $\mu(x_3, e, p)$ the probabilities for these outcomes to occur, performing the measurement $e$, the entity being in state $p$. The construction leading to the representation theorem takes then place in $\mathbb{R}^3$. We have represented the canonical basis vectors $h_1=(1,0,0)$, $h_2=(0,1,0)$ and $h_3=(0,0,1)$ of $\mathbb{R}^3$ in Fig. \ref{threedimensional}, and also drawn the simplexes $S_3(e)$, $A_1(e,p)$, $A_2(e,p)$ and $A_3(e,p)$. 
\begin{figure}[H] 
\centerline {\includegraphics[scale=1.1]{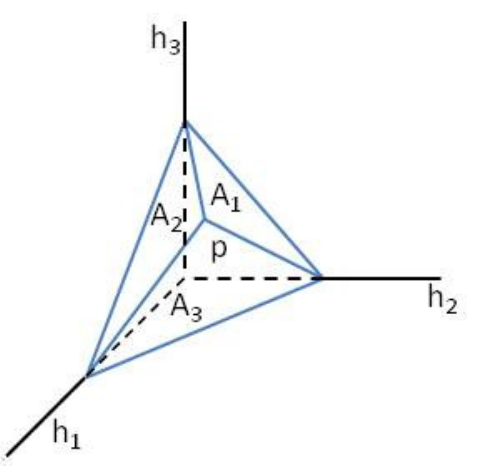}}
\caption{A simple 3-dimensional picture showing the construction needed for the representation theorem.} \label{threedimensional}
\end{figure}
\noindent
We now introduce the vector $v(e,p)=\mu(x_1,e,p)h_1+\mu(x_2,e,p)h_2+\mu(x_3,e,p)h_3=(\mu(x_1,e,p),\mu(x_2,e,p),$ $\mu(x_3,e,p))$. We have that $A_1(e,p)$, $A_2(e,p)$ and $A_3(e,p)$ are the convex closures of $\{v(e,p), h_2, h_3 \}$, $\{ h_1, v(e,p), h_3 \}$ and $\{ h_1, h_2, v(e,p) \}$, respectively. Then, let the point $\lambda$ belonging to the simplex $S_3(e)$ be defined as $\lambda=\lambda_1 h_1+\lambda_2 h_2+\lambda_3 h_3=(\lambda_1,\lambda_2,\lambda_3)$, with $0 \le \lambda_1,\lambda_2,\lambda_3 \le 1$ and $\lambda_1+\lambda_2+\lambda_3=1$. Finally, measurement $e$ gives outcome $x_j$ with certainty when $S$ is in state $p$ if and only if $\lambda \in A_j(e,p)$, and $\mu(x_j, e, p)=m(A_j(e,p))/m(S_3(e))$. 

Coming to the quantum representation in $\mathbb{C}^3$, we introduce the orthogonal projection operators
\begin{equation}
M_{1}=
\left (
\begin{array}{ccc} 
1 & 0 & 0 \\ 
0 & 0 & 0 \\
0 & 0 & 0
\end{array}\right )
\quad 
M_{2}=
\left (
\begin{array}{ccc} 
0 & 0 & 0 \\ 
0 & 1 & 0 \\
0 & 0 & 0
\end{array}\right )
\quad 
M_{3}=
\left (
\begin{array}{ccc} 
0 & 0 & 0 \\ 
0 & 0 & 0 \\
0 & 0 & 1
\end{array}\right ),
\end{equation}
and the vector $w(e,p)=(\sqrt{\mu(x_1,e,p)}e^{i\alpha},\sqrt{\mu(x_2,e,p)}e^{i\beta}, \sqrt{\mu(x_3,e,p)}e^{i\gamma})$. Then, we have $\mu(x_j, e, p)=\langle w(e,p)|M_{j}|w(e,p)\rangle=\parallel M_{j}|w(e,p)\rangle\parallel^{2}$, $j=1,2,3$.

\section{Interference and superposition\label{interference}}
The results obtained in the previous sections can be applied at once to a typical quantum phenomenon, namely, interference \cite{young1802,debroglie1923,schrodinger1926,debroglie1928,jonsson1961,feynman1965,jonsson1974,ArndtNairzVos-AndreaeKellervanderZouwZeilinger1999}.

To investigate a situation of interference in our general modeling scheme we introduce a measurement $e$ with an outcome set $\{x_1, \ldots, x_j, \ldots, x_n\}$, and sets of probabilities $\{P(x_j, e, p)\vert\ j=1\ldots n\}$, $\{P(x_j, e, q)\vert\ j=1\ldots n\}$ and $\{P(x_j, e, r)\vert\ j=1\ldots n\}$ with respect to states $p$, $q$ and $r$. We now wonder how the probabilities with respect to $r$ are related to the probabilities with respect to $p$ and the ones with respect to $q$, in case $r$ is a superposition state of $p$ and $q$.  
Let us consider this situation. Hence, we suppose that the vector $w(e,r)$ is a linear combination of the vectors $w(e,p)$ and $w(e,q)$, more specifically $w(e,r)=ae^{i\alpha}w(e,p)+be^{i\beta}w(e,q)$. From this follows that
\begin{eqnarray}
w(e,r)&=&\sum_{j=1}^n\sqrt{P(x_j, e, r)}e^{i\alpha(e,r)_j}h_j \nonumber \\
&=&\sum_{j=1}^n(ae^{i\alpha}\sqrt{P(x_j, e, p)}e^{i\alpha(e,p)_j}+be^{i\beta}\sqrt{P(x_j, e, q)}e^{i\alpha(e,q)_j})h_j
\end{eqnarray}
and hence
\begin{equation}
\sqrt{P(x_j, e, r)}e^{i\alpha(e,r)_j}=\sqrt{a^2P(x_j, e, p)}e^{i(\alpha(e,p)_j+\alpha)}+\sqrt{b^2P(x_j, e, q)}e^{i(\alpha(e,q)_j+\beta)} \nonumber
\end{equation}
which leads to
\begin{eqnarray}
&&P(x_j, e, r)=(\sqrt{a^2P(x_j, e, p)}e^{-i(\alpha(e,p)_j+\alpha)}+\sqrt{b^2P(x_j, e, q)}e^{-i(\alpha(e,q)_j+\beta)}) \cdot \nonumber \\
&&(\sqrt{a^2P(x_j, e, p)}e^{i(\alpha(e,p)_j+\alpha)}+\sqrt{b^2P(x_j, e, q)}e^{i(\alpha(e,q)_j+\beta)}) \nonumber \\
&&=a^2P(x_j, e, p)+b^2P(x_j, e, q) \nonumber \\ \label{interference1}
&&+2ab\sqrt{P(x_j, e, p)P(x_j, e, q}\cos(\alpha(e,p)_j-\alpha(e,q)_j+\alpha-\beta) .
\end{eqnarray}
The third term of Eq. (\ref{interference1}) is the \emph{interference term}. If this term is different from zero, which is generally the case, the vector $v(e,r)$ is not located on the line segment between the vectors $v(e,p)$ and $v(e,q)$.


\begin{thebibliography}{99}
\setlength{\itemsep}{-1mm}
\bibitem{aerts1982} Aerts, D.: Example of a macroscopical situation that violates Bell inequalities. Lett. N. Cim. 34, 107--111 (1982)

\bibitem{aerts1986} Aerts, D.: A Possible Explanation for the Probabilities of Quantum Mechanics. J. Math. Phys. 27, 202--210 (1986)

\bibitem{aertsdurtgribvanbogaertzapatrin1993} Aerts, D., Durt, T., Grib, A., Van Bogaert, B. and Zapatrin, A.: Quantum structures in macroscopical reality. Int. J. Theor. Phys. 32, 489--498 (1993)

\bibitem{aerts2009} Aerts, D.: Quantum Structure in Cognition. J. Math. Psych. 53, 314--348 (2009)

\bibitem{aerts2009b} Aerts, D.: Quantum Particles as Conceptual Entities: A Possible Explanatory Framework for Quantum Theory. Found. Sci. 14, 361--411 (2009)

\bibitem{aerts2010} Aerts, D.: Interpreting Quantum Particles as Conceptual Entities. Int. J. Theor. Phys. 49, 2950--2970 (2010)

\bibitem{aerts2010b} Aerts, D.: A Potentiality and Conceptuality Interpretation of Quantum Physics. Philosophica 83, 15--52 (2010)

\bibitem{aertsgabora2005a} Aerts, D., Gabora, L.: A Theory of Concepts and Their Combinations I\&II: Kybernetes, 34, 167--191; 192--221 (2005)

\bibitem{aerts2002} Aerts, D.: Being and Change: Foundations of a Realistic Operational Formalism. In: Aerts, D., Czachor, M., Durt, T. (eds.) Probing the Structure of Quantum Mechanics: Nonlinearity, Nonlocality, Probability and Axiomatics, pp. 71--110. World Scientific, Singapore (2002)

\bibitem{aertsgabora2002} Gabora, L., Aerts, D.: Contextualizing Concepts Using a Mathematical Generalization of the Quantum Formalism. J. Exp. Theor. Art. Int. 14, 327--358 (2002)


\bibitem{nelson2007} Nelson, D.L., Entangled Associative Structures and Context. In: Bruza, P., Lawless, W., van Rijsbergen, K., Sofge, D. (eds.) Proceedings of the Association for the Advancement of Artificial Intelligence (AAAI) Spring Symposium 8: Quantum Interaction, March 26--28, 2007, Stanford University, Stanford (2007)

\bibitem{gaboraroschaerts2008} Gabora, L., Rosch, E., Aerts, D.: Toward an Ecological Theory of Concepts. Ecol. Psych. 20, 84--116 (2008)

\bibitem{flenderkittobruza2009} Flender, C., Kitto, K., Bruza, P.:  Beyond Ontology in Information Systems. Quantum Interaction, Lecture Notes in Computer Science 5494, 276--288 (2009)

\bibitem{gaboraaerts2009} Gabora, L., Aerts, D.: A Model of the Emergence and Evolution of Integrated Worldviews. J. Math. Psych. 53, 434--451 (2009).

\bibitem{dhooghe2010} D'Hooghe, B.: The SCOP-formalism: An Operational Approach to Quantum Mechanics. AIP Conference Proceedings 1232, pp. 33--44 (2010)

\bibitem{aertsczachorsozzo2010} Aerts, D., Czachor, M., Sozzo, S.: A Contextual Quantum-based Formalism for Population Dynamics. Proceedings of the AAAI Fall Symposium (FS-10-08), Quantum Informatics for Cognitive, Social, and Semantic Processes, pp. 22--25, (2010)
 
\bibitem{velozgaboraeyjolfsonaerts2011}
Veloz, T., Gabora, L., Eyjolfson, J., Aerts, D.: Toward a Formal Model of the Shifting Relationship Between Concepts and Contexts During Associative Thought. In: Song, D., Melucci, M., Frommholz, I. (eds.) Proceedings of QI 2011-Sixth International Symposium on Quantum Interaction. LNCS, vol. 7052, pp. 25--34. Springer, Berlin, Heidelberg (2011)

\bibitem{aertssozzo2011} Aerts, D., Sozzo, S.: Quantum Structure in Cognition: Why and How Concepts Are Entangled. In: Song, D., Melucci, M., Frommholz, I. (eds.) Proceedings of QI 2011-Fifth International Symposium on Quantum Interaction. LNCS, vol. 7052, pp. 116--127. Springer, Berlin, Heidelberg (2011)

\bibitem{young1802} Young, T.: On the Theory of Light and Colours. Phil. Trans. Roy. Soc. 92, 12--48 (1802). Reprinted in part in: Crew, H. (ed.) The Wave Theory of Light, New York (1990)

\bibitem{debroglie1923} de Broglie, L.: Ondes et Quanta. Comptes Rendus 177, 507--510 (1923)

\bibitem{schrodinger1926} Schr\"odinger, E.: Quantizierung als Eigenwertproblem(Erste Mitteilung). Ann. Phys. 79, 361--376 (1926)

\bibitem{debroglie1928} de Broglie, L.: La Nouvelle Dynamique des Quanta. In: Proceedings of the Solvay Conference-1928, Electrons et Photons, pp. 105--132 (1928)

\bibitem{jonsson1961} J\"onsson, C.: Elektronen Interferenzen an Mehreren K\"unstlich Hergestellten Feinspalten. Zeit. Phys. 161, 454--474 (1961)

\bibitem{feynman1965} Feynman, R.P.: The Feynman Lectures on Physics. Addison--Wesley, New York (1965)

\bibitem{jonsson1974} J\"onsson, C.: Electron Diffraction at Multiple Slits. Am. J. Phys. 4, 4--11 (1974)

\bibitem{ArndtNairzVos-AndreaeKellervanderZouwZeilinger1999} Arndt, M., Nairz, O., Vos-Andreae, J., Keller, C., van der Zouw, G., Zeilinger, A.: Wave-particle Duality of $C_{60}$ Molecules. Nature 401, 680--682 (1999)


\bibitem{aerts1993} Aerts, D.: Quantum Structures due to Fluctuations of the Measurement Situations. Int. J. Theor. Phys. 32, 2207--2220 (1993)

\bibitem{aerts1994} Aerts, D.: Quantum Structures, Separated Physical Entities and Probability. Found. Phys. 24, 1227--1259 (1994)

\bibitem{aertsaerts1997} Aerts, D., Aerts, S.: The Hidden Measurement Formalism: Quantum Mechanics as a Consequence of Fluctuations on the Measurement. In: Ferrero, M., van der Merwe, A. (eds.) New Developments on Fundamental Problems in Quantum Physics, pp. 1--6. Springer, Dordrecht (1997)

\bibitem{aertsaertscoeckedhooghedurtvalckenborgh1997} Aerts, D., Aerts, S., Coecke, B., D'Hooghe, B., Durt, T., Valckenborgh, F.: A Model with Varying Fluctuations in the Measurement Context. In: Ferrero, M., van der Merwe, A. (eds.) New Developments on Fundamental Problems in Quantum Physics, pp. 7--9. Springer, Dordrecht (1997)

\bibitem{aertssven2002} Aerts, S.: Hidden Measurements from Contextual Axiomatics. In: Aerts, D., Czachor, M., Durt, T. (eds.) Probing the Structure of Quantum Mechanics: Nonlinearity, Nonlocality, Probability and Axiomatics, pp. 149--164. World Scientific, Singapore (2002) 

\bibitem{aertssven2005} Aerts, S.: The Born Rule from a Consistency Requirement on Hidden Measurements in Complex Hilbert Space. Int. J. Theor. Phys. 44, 999--1009 (2005)

\bibitem{aertssozzo2012b} Aerts, D., Sozzo, S.: Entanglement of Conceptual Entities in QMod Theory. Submitted to the Proceedings of QI 2012-Sixth International Symposium on Quantum Interaction (2012) 

\bibitem{aertsgaborasozzo2012} Aerts, D., Gabora, L. and Sozzo, S.: 
How Concepts Combine: A Quantum Theoretic Modeling of Human Though. Accepted for publication in Topics in Cognitive Science (2012).

\end{thebibliography}
\end{document}